\documentclass[conference]{IEEEtran} 

\usepackage[latin1]{inputenc}  
\usepackage{fontenc}          

\usepackage{textcomp}

\usepackage[dvips]{graphicx} 
\usepackage{epsfig}
\usepackage{amssymb}
\usepackage{amsmath}
\usepackage{amsfonts}

\usepackage{latexsym}
\usepackage{euscript, stmaryrd, bigstrut}

\usepackage{float}

\usepackage[latin1]{inputenc}   
\usepackage{fontenc}            

\usepackage[ruled,lined,linesnumbered]{./algorithm2e}   

\usepackage{listings}

\usepackage{url}

\usepackage{wrapfig}

       \usepackage{multirow}
       \usepackage{hhline}

\usepackage[autolanguage,np]{numprint}  

\newtheorem{theorem}{Theorem}
\newtheorem{definition}{Definition}

\newcommand{\myparagraph}[1]{\smallskip \noindent\textbf{#1}}  

\newfloat{Example}{htbp}{loe}  
\floatname{Example}{{\bf Ex.}} 

\newcommand{\card}[1]{|#1|}

\newcommand{\mydef}{\triangleq}

\newcommand{\labfun}{\psi}  

\newcommand{\lit}[1]{\bar{#1}}   

\newcommand{\Pg}{P}   
\newcommand{\loc}{loc}   
\newcommand{\pred}{\phi}   
\newcommand{\Lab}{L}   
\newcommand{\lab}{l}   
\newcommand{\TS}{TS}   
\newcommand{\dt}{t}  
\newcommand{\covers}{\leadsto} 
\newcommand{\fullcovers}{\leadsto} 
\newcommand{\C}{{\bf C}} 

\newcommand{\LC}{{\bf LC}}  
\newcommand{\IC}{{\bf IC}}  
\newcommand{\DC}{{\bf DC}}  
\newcommand{\CC}{{\bf CC}}  
\newcommand{\DCC}{{\bf DCC}}  
\newcommand{\MCC}{{\bf MCC}}  
\newcommand{\WM}{{\bf WM}}  
\newcommand{\SM}{{\bf M}}  

\newcommand{\Paths}{Paths}  

\newcommand{\pathcrawler}{{\sc PathCrawler}}

\begin{document}

\IEEEoverridecommandlockouts

\title{Efficient Leverage of Symbolic ATG Tools \\ to Advanced Coverage Criteria}
%

\author{\IEEEauthorblockN{Sébastien Bardin\IEEEauthorrefmark{1},
Nikolai Kosmatov\IEEEauthorrefmark{1} and 
Fran\c{c}ois Cheynier\IEEEauthorrefmark{1}}

\medskip 

\IEEEauthorblockA{\IEEEauthorrefmark{1} CEA, LIST, Saclay, France \\ Email: {\tt first.name@cea.fr}}
}

\author{\IEEEauthorblockN{Sébastien Bardin}
\IEEEauthorblockA{ CEA, LIST\\
        Gif-sur-Yvettes\\ 
        91191, France\\ 
 \texttt{\small sebastien.bardin@cea.fr}
}
\and
\IEEEauthorblockN{Nikolai Kosmatov}
\IEEEauthorblockA{CEA, LIST\\
        Gif-sur-Yvettes\\ 
        91191, France\\ 
 \texttt{\small nikolai.kosmatov@cea.fr}
}
\and
\IEEEauthorblockN{Fran\c{c}ois Cheynier}
\IEEEauthorblockA{CEA, LIST\\
        Gif-sur-Yvettes\\ 
        91191, France\\ 
 \texttt{\small francois.cheynier@gmail.com}
}
}


\date{\ }

\maketitle
\begin{abstract}
 Automatic test data generation (ATG) is a major topic in software engineering. 
In this paper, we seek to bridge the gap between the coverage criteria supported by symbolic ATG tools     
and the most advanced coverage criteria found in the literature. 
 We define a new testing criterion, label coverage, and prove it to be both expressive and 
 amenable to efficient automation.     
We propose   several innovative techniques resulting in an effective black-box support for label coverage, while a direct approach induces an exponential blow-up 
of the search space. 
Initial experiments show that ATG for label coverage can be achieved 
 at a reasonable cost and that our optimisations yield  very  significant savings. 
%
%
\end{abstract}












\begin{IEEEkeywords}
Testing, symbolic execution, coverage criteria
\end{IEEEkeywords}

\section{Introduction}  



\myparagraph{Context and problem.} Automatic test data generation (ATG) is a major concern in software engineering and program analysis. 
Recent progress in automated theorem proving led to significant improvements of  symbolic approaches for white-box ATG, such as Dynamic Symbolic Execution  (DSE) \cite{GKS-05,SMA-05,WMM-04a}.  
Tools have been developed \cite{BH-11,CDE-08,CGPDE-06,GLM-08,TH-08} and  impressive case-studies have been carried 
out \cite{CDE-08,CGPDE-06,GLM-12}.  
%
%

 DSE follows mostly an exhaustive exploration of the path space of the program under test, covering all  execution paths up to a given bound. 
While this ``all-path coverage'' criterion proves successful in some contexts, it is well known that the resulting test suite can still miss bugs related to data rather than control. 
%
Moreover, standard DSE does not support coverage objectives defined over artifacts not explicitly present in the source code, such as  multiple-condition coverage.  
%
%

On the other hand, many coverage criteria have been defined along the years \cite{AO-08}, ranging from control-flow or data-flow  criteria to mutations \cite{DLS-78},   
input domain partitions and MCDC. But only very few of them are incorporated inside DSE tools, while they could efficiently guide test generation.



\myparagraph{Goal.}  Our main objective is to bridge the gap between coverage criteria supported by symbolic ATG tools, especially DSE,  
and the most advanced coverage criteria found in the literature. 
Recent works aim at leveraging DSE to mutation testing \cite{PM-10,PM-11,PMK-10}  or  improving DSE bug-detection abilities 
by making explicit run-time error conditions~\cite{CKGJ-10,GLM-08b,KA-96}. Interestingly, these approaches are mainly based on 
instrumentation and allow for black-box reuse of existing ATG tools. However, they come at a high price since they induce a blow-up of the path space and a significant overhead.    
We follow the same general line, emphasising black-box reuse as much as possible.    
However, we focus on  two main points mostly left unaddressed:  we want to characterize  which kind of coverage criteria can be  
supported by DSE-like techniques, and we want to support them efficiently.   

\myparagraph{Approach.} We define \textit{label coverage}, a new testing criterion which appears to be both expressive and amenable to efficient automation.  
Especially, it turns out that {\it DSE can be extended for label coverage with only a slight overhead}. 
Labels are predicates attached to program instructions through a labelling function. A label is covered if a test execution reaches it and satisfies the predicate.  
This idea underlies former work on the subject \cite{CKGJ-10,GLM-08b,KA-96,PMK-10}. We  generalize these results and propose ways of taming the potential blow-up.  
Especially, we introduce a {\it tight instrumentation}, where ``tight'' is made precise in the paper,  and a strong coupling of DSE and label coverage named  {\it iterative label deletion}. 
The combination results in an effective support for label coverage in DSE. 
Interestingly, both techniques can be implemented in black-box.

\myparagraph{Contribution.} Our main contributions are the following: 
\begin{itemize} 
\item 
We show that label coverage is expressive enough to faithfully emulate many standard coverage criteria, from decision or condition coverage to input domain coverage 
and a substantial subset of weak mutations (the side-effect free fragment, Theorem~\ref{thm:wm-lc}). 

\item We formally characterise the properties of  direct instrumentation for label coverage. 
 We show that the instrumentation is sound (w.r.t.~label coverage) and  
leads to very efficient coverage score computation. However,  it is very ineffective for any analysis working through path exploration, 
as it yields an exponential increase as well as  a ``complexification'' of the path space (Theorem~\ref{thm:non-tight}).

\item  We propose \textit{tight instrumentation} and \textit{iterative label deletion} as ways of taming this complexity blow-up. 
Tight instrumentation  yields   
only a linear growth  of the path space without any complexification (Theorem~\ref{thm:tight-tight}).  
Both techniques are orthogonal and allow for a significant speed-up. Moreover, they can both be implemented either through dedicated DSE algorithms or in a black-box manner.

\item We have implemented these results inside a DSE tool \cite{WMM-04a}.  Initial experiments on small benchmarks show that ATG for label coverage can be achieved 
 at a reasonable cost  w.r.t.~the usual (all-path) DSE approach,  while our optimisations yield  very  significant reductions of both  search space and computation time  
compared to  direct instrumentation.  


\end{itemize}

\noindent As a whole, label coverage forms the basis of a very generic and convenient framework for test automation, providing a powerful specification mechanism for test objectives 
and featuring  
   efficient integration into symbolic ATG techniques as well as  cheap coverage score computation. Moreover, static analysis techniques can also be used 
  directly on the instrumented programs in order to  detect uncoverable labels, as was proposed for mutation testing~\cite{OC-94}.


This work bridges part of the gap between symbolic ATG techniques and coverage criteria. On the one hand, we show that  
DSE techniques  can be cheaply extended to more advanced testing criteria, such as side-effect free weak mutations. 
%
 On the other hand, we identify a large subclass of weak mutations amenable to  efficient automation, both in terms of ATG and mutation score computation.  

\myparagraph{Outline.} The remaining part of the paper is structured as follows. After presenting  basic notation (Section~\ref{sec:notations}), we define labels and explore their expressiveness (Section~\ref{sec:labels}). We then focus on automation. The direct instrumentation is defined and studied (Section~\ref{sec:naive}).  Afterwards, we describe our own approach to label-based ATG 
 (Section~\ref{sec:smart})  and first experiments are presented (Section~\ref{sec:experiments}). Finally, we sketch a highly automatized testing framework based on labels (Section~\ref{sec:framework}), discuss related work 
 (Section~\ref{sec:related}) and give a conclusion (Section~\ref{sec:conclusion}).

\section{Background} \label{sec:notations}

\subsection{Notation}

 Given a program $\Pg$ over a vector of input variables $V$ taking values in some domain $D$, 
a test data $\dt$ for $\Pg$ is any valuation of $V$, i.e.~$\dt \in D$.  The execution of $\Pg$ over $\dt$, denoted $\Pg(\dt)$, is  formalized as a path (or run) $\sigma \mydef\ (\loc_1,S_1) \ldots (\loc_n,S_n)$, where  the $\loc_i$ denote control-locations (or control-points, or simply locations) of $\Pg$ and the $S_i$ denote 
the successive internal states of $\Pg$  ($\approx$ valuation of all global and local variables as well as memory-allocated structures) before the execution of  each $\loc_i$.  
A test data $\dt$ reaches a specific location $\loc$ with internal state $S$, denoted  $\dt \covers_P (\loc, S)$, if $\Pg(\dt)$ is of the form $\sigma_1 \cdot (\loc,S) \cdot \sigma_2$.   
A test suite $\TS$ is a finite set of test data.

Given a test objective {\bf c}, we write  $\dt \covers_P \text{\bf c}$ if test data $\dt$ covers {\bf c}. 
We extend the notation for a test suite $\TS$ and a set of test objectives {\bf C}, writing   $\TS \covers_P \text{\bf C}$ 
when  
for any $\text{{\bf c}} \in \text{{\bf C}}$, there exists $\dt \in \TS$ such that $ \dt \covers_P \text{\bf c}$.

The above definitions are generic and leave the exact definition of ``covering'' to the considered  testing criterion.  
%
%
For example,  
test objectives derived from the Decision Coverage criterion are of the form $\text{\bf c} \mydef\ (\loc, \text{\tt cond})$ or $\text{\bf c} \mydef\  (\loc, \text{\tt !cond})$, where {\tt cond} is 
 the condition of the branching instruction at location $\loc$.   Here,  
  $\dt \covers_P \text{\bf c}$ if   $\dt$ reaches some  $(\loc, S)$ where {\tt cond} evaluates to {\it true} (resp.~{\it false}) in $S$.

\subsection{DSE in brief} 

We  remind here  a few basic facts about Symbolic Execution (SE) \cite{Kin-70} and  Dynamic Symbolic Execution (DSE) \cite{GKS-05,SMA-05,WMM-04a}.  
Let us consider a program under test $\Pg$ with input variables $V$  over domain $D$ and a path $\sigma$ of $\Pg$.  
The key insight of SE is that it is possible in many cases  to compute a {\it path predicate} $\pred_{\sigma}$ for $\sigma$ such that for any input valuation  
 $\dt \in D$, we have:  $ \dt$ satisfies $\pred_{\sigma}$  iff $\Pg(\dt)$ covers $\sigma$.  In practice, path predicates are often under-approximated  and only the  
  left-to-right implication holds, which is already fine  for testing:  SE  outputs a set of pairs $(\dt_i,\sigma_i)$ such that each $\dt_i$  is ensured to cover the corresponding 
$\sigma_i$. Hence, SE is  {\it sound} from a testing point of view.   
DSE  enhances SE by interleaving concrete and symbolic executions. The dynamically collected information  can   help the symbolic step, for example by 
suggesting relevant approximations.

A simplified view of SE is depicted in Algorithm~\ref{algo:se}. While high level, it is sufficient to understand the rest of the paper.  
We assume that the set of paths of $\Pg$, denoted  $\Paths(\Pg)$, is finite. In practice, DSE  tools enforce this assumption through a bound on path lengths. 
We assume the availability of a procedure for  path predicate computation (with predicates in some theory $T$), as well as the availability   of 
  a solver taking a formula $\phi \in T$ and returning either {\it sat} with a solution $\dt$ or {\it unsat}. All  DSE tools rely on such procedures.   The algorithm builds iteratively a test suite $\TS$  by exploring 
all paths  from $\Paths(\Pg)$.

\begin{algorithm}
\KwIn{a program $\Pg$ with finite set of paths  $\Paths(\Pg)$ }
\KwOut{$\TS$, a set of pairs $(\dt,\sigma)$ such that $\Pg(\dt) \covers_{\Pg} \sigma$ }

 $\TS$ := $\emptyset$\; 
 $S_{paths}$ :=   $\Paths(\Pg)$\;

 \While{$S_{paths} \neq \emptyset$}{
    choose $\sigma \in S_{paths}$; $S_{paths}$ := $S_{paths} \backslash \{\sigma$\}  \;

    compute path predicate $\pred_{\sigma}$ for $\sigma$ \;

\Switch{solve($\pred_{\sigma}$)}{

\lCase{sat($\dt$): }{ $\TS$ := $\TS \cup  \{(\dt,\sigma)\}$}

\lCase{unsat: }{ skip }

}
 }

\Return{$\TS$}\;

\caption{Symbolic Execution algorithm}\label{algo:se}
\end{algorithm}

The major issue here  is that SE and DSE must in some ways explore all $\Paths(\Pg)$. Advanced tools explore this set lazily,  
yet they still have to crawl it. 
Therefore, the size of $\Paths(\Pg)$, denoted  $\card{\Paths(\Pg)}$,  is one of the two major bottlenecks of SE and DSE, 
the other one being the average cost of solving path predicates.  

Bounded model checking (BMC) \cite{CKL-04} is sensitive to the same parameters, as it amounts to building a large formula encompassing all paths up to a given length. 
Especially, more paths yield  larger formulas with more $\vee$-operators.

\section{Label coverage}  \label{sec:labels}


     \subsection{Definitions}

Given a program $\Pg$, 
a \textit{label} $\lab$ is a pair $(\loc,\pred)$ where $\loc$ is a location of $\Pg$ 
and $\pred$ is a predicate obeying the following rules:  
\begin{itemize}
\item $p$ contains only variables and expressions well-defined in  $\Pg$ at location $\loc$; 
\item $p$ contains no side-effect expressions. 
\end{itemize}

An {\it annotated program} is a pair $\langle \Pg,\Lab \rangle$  where $\Lab$ is a set of labels defined over $\Pg$.  
A test data $\dt$ covers $\lab \mydef\ (\loc,\pred)$, denoted $\dt \covers_{\Pg} \lab$, if $\dt$ covers some $(\loc, S)$   with $S$ satisfying predicate  $\pred$.   
%
%
The label coverage testing criterion will be denoted by \LC.

For simplicity, we consider in the rest of the paper  {\it normalized programs}, i.e.~programs such that no side-effect occurs in 
any condition of a branching instruction. This is not a severe restriction since any (well-defined) program $\Pg_1$ can be   
rewritten into a normalized program $\Pg_2$, using intermediate variables to evaluate the  
  side-effect prone  conditions outside the branching instruction. For example, {\tt if (x++ $<$= y \&\& e==f)} becomes 
{\tt b = (x++ $<$= y); if (b \&\& e==f)} or {\tt tmp = x++; if (tmp $<$= y \&\& e==f)}. 
Notice that similar transformations are automatically performed 
by the Cil library \cite{Necula02cil} 
frequently used by DSE tools for C programs \cite{SMA-05,WMM-04a}.

\subsection{Expressiveness of label coverage}


We seek to characterize the power of the \LC\ testing criterion.  
We prove in Theorem~\ref{thm:standard-lc} and Theorem~\ref{thm:wm-lc}  that labels allow  to simulate many standard testing criteria, including \DC\ (decision coverage),
\MCC\ (multiple-condition coverage) and a large subset of \WM\ (weak mutations). 
A key notion is  that of {\it labelling function}. A labelling function $\labfun$ maps a program $\Pg$ into an annotated program $\langle \Pg, \Lab \rangle$.  
We write $\LC_{\labfun}$ to denote  the coverage of  labels defined by $\labfun$.

\begin{definition}\label{def:simulation}
 A coverage criterion \C\ {\it can be simulated by} \LC\  if there exists a labelling function $\labfun$ such that for any program $\Pg$, 
a test suite $\TS$ covers \C\ iff  $\TS$ covers \LC$_{\labfun}$.
\end{definition}


We show first how \LC\ can simulate basic graph and logic coverage criteria. 
We consider the following coverage criteria: instruction coverage \IC, 
decision coverage \DC, (simple) condition coverage \CC, decision-condition coverage \DCC\   and multiple-condition coverage \MCC. 
The basic idea is to introduce in $\Pg$ labels based on branching predicates and their atomic conditions. 
An  example for \CC\ is depicted in Figure~\ref{fig:lc-simulates-standard}, where  additional labels (right) 
enforce coverage of the two atomic conditions {\tt x==y} and {\tt a<b}. 

\begin{figure}[htbp]
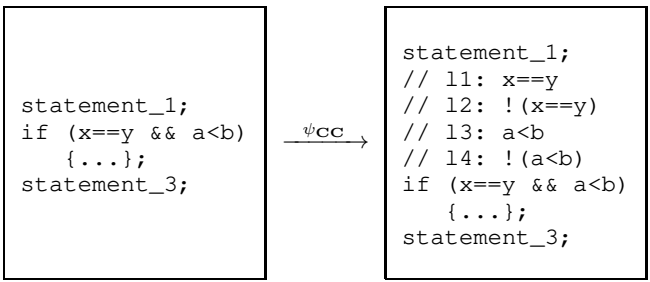


\begin{center} \small \tt
\setlength{\columnseprule}{0cm}
\begin{tabular}{|c|c|c|} 
\cline{1-1}\cline{3-3}
  && \\ 
\begin{lstlisting}
statement_1;
if (x==y && a<b) 
   {...};  
statement_3; 
\end{lstlisting}
&$ \xrightarrow{\ \labfun_{\CC} \ }$&
\begin{lstlisting}
statement_1;
// l1: x==y 
// l2: !(x==y)
// l3: a<b
// l4: !(a<b)
if (x==y && a<b) 
   {...}; 
statement_3;
\end{lstlisting}\\
  && \\ 
\cline{1-1}\cline{3-3}
\end{tabular}
\end{center}

\caption{Simulating \CC\ with labels}\label{fig:lc-simulates-standard}
\end{figure}

\begin{theorem}\label{thm:standard-lc} 
The coverage criteria  \IC, \DC, \CC, \DCC\ and \MCC\ can be simulated by \LC. 
\end{theorem}

\begin{proof}
We need to define a suitable labelling function  for any of the considered coverage criteria. 
For \IC,  we choose the labelling function $\labfun_{\IC}(\Pg)$ adding all labels of the form $(\loc,true)$,  where $\loc$ is any location of $\Pg$. 
Given a test suite $\TS$, 
$\TS \covers_{\Pg} \IC$ iff $\TS$ can reach any $\loc$ of $\Pg$  iff $\TS$ covers any $(\loc,true)$ iff $\TS \covers_{\Pg} \LC_{\labfun_{\IC}}$.  We  conclude that 
\IC\ can be simulated by \LC. 

Other criteria are handled similarly.  
%
The labelling function $\labfun_{\DC}$ adds the set of all $(\loc,\pred)$ and $(\loc,\neg \pred)$, where  $\loc$ contains a conditional statement 
with condition $\pred$. 
%
The labelling function $\labfun_{\CC}$ adds the set of all $(\loc,a_i)$ and $(\loc,\neg a_i)$, where   $\loc$ contains a conditional statement 
whose atomic conditions  are exactly the $a_i$.  
The labelling function  $\labfun_{\DCC}$ adds the union of  $\labfun_{\DC}$ and $\labfun_{\CC}$. 
The labelling function $\labfun_{\MCC}$ adds  the set of all $(\loc, \bigwedge_i \lit{a_i})$, where the $a_i$ are atomic conditions and $\lit{a_i}$ denotes either $a_i$ or $\neg a_i$.  
\end{proof}

\myparagraph{Weak mutations.} We now consider a more involved testing criterion, namely weak mutations. In mutation testing \cite{DLS-78}, test objectives consist 
of mutants, i.e.~slight syntactic modifications of the program under test. In the strong mutation setting \SM, a mutant $M$ is covered (or killed) by a test data $\dt$   
if the  output of $P(\dt)$ differs from the output of $M(\dt)$. In the weak mutation setting \WM\ \cite{Howden-82},  a  mutant $M$ is covered  by  $\dt$, denoted $\dt \covers_{\Pg} M$,   
if  the internal states of $P(\dt)$ and $M(\dt)$ differ from each other right after the mutated location (cf.~Figure~\ref{fig:strong-weak}). \SM\ is  a
 powerful testing criterion in practice~\cite{ABL-05,OU-01}. 
While less powerful in theory, \WM\ appears to be 
almost equivalent to \SM\ in practice \cite{OL-94}.

\begin{figure}[htbp]
\begin{center}
 \includegraphics[width=\columnwidth]{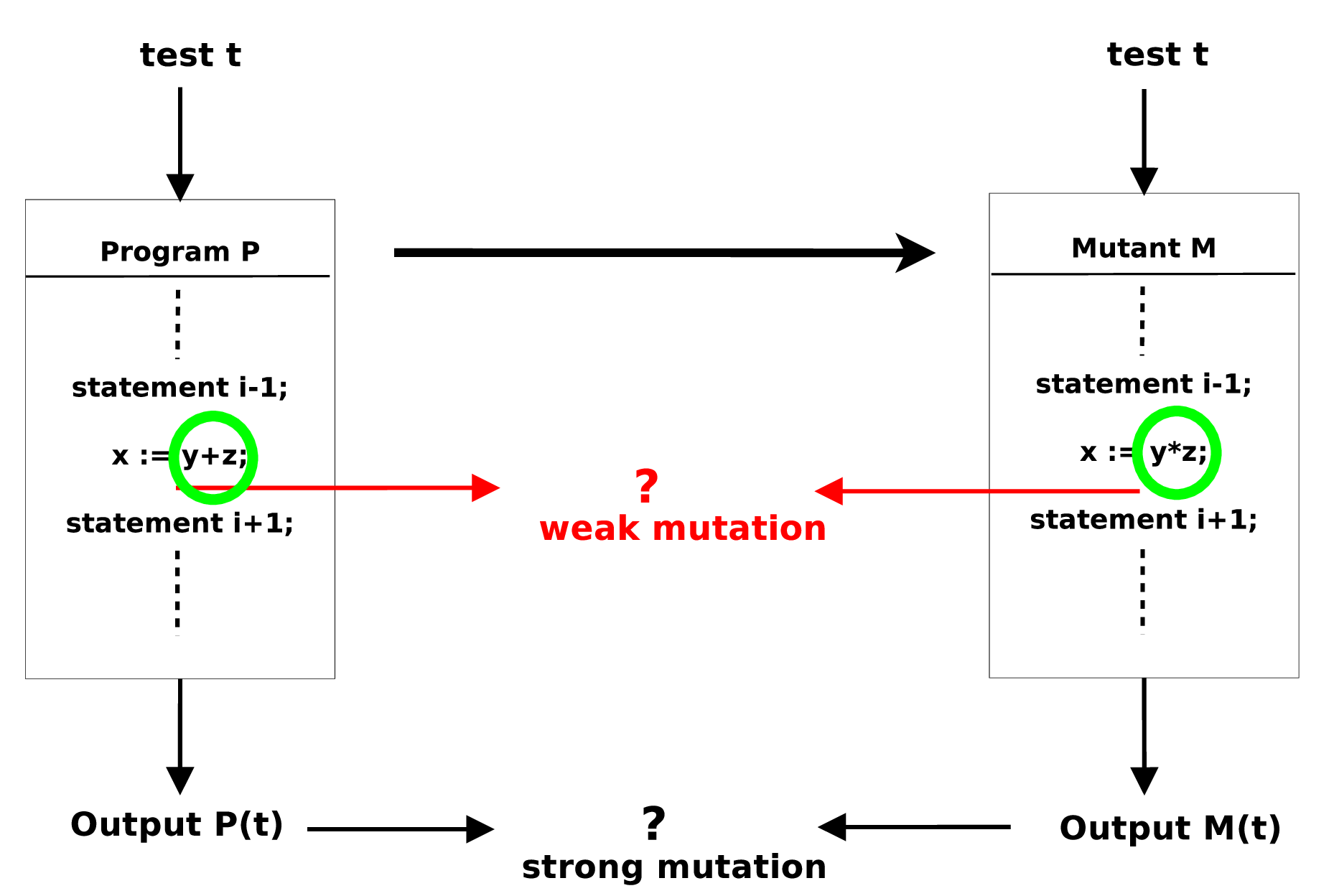}
\end{center}
\caption{Strong and weak mutations}\label{fig:strong-weak}
\end{figure}

We show hereafter that a substantial part of \WM\ can be simulated by \LC. 
First we need a few more definitions. Mutation testing is parametrised by a set of {\it mutation operators} $O$.  
A mutation operator $op \in O$ is a function mapping a program $\Pg$ into a finite set of well-defined programs (mutants), such that $\Pg$ differs from each mutant $M$ in only one 
location (atomic mutation). We denote \WM$_{O}$ the weak mutation criterion resticted to mutants created through operators in $O$.  
We consider that mutations can affect either a lhs value, an expression or a condition.  
This is a very generic model of mutations, encompassing all standard operators~\cite{AO-08}.  
%
%
%
%
%
%
%
Finally, we restrict ourselves to mutation operators neither affecting nor introducing side-effect expressions (including calls to side-effect prone functions).  
We refer to such operators as {\it side-effect free mutation operators}.

\begin{theorem}\label{thm:wm-lc} 
For any finite set $O$ of side-effect free mutation operators, 
\WM$_{O}$ can be simulated by \LC.  
\end{theorem}

\begin{proof}
For simplicity, let us consider first a single mutation operator $op \in O$. 
The main idea is to introduce {\it one label for each mutant} created by $op$. The label  encodes the necessary and sufficient conditions to distinguish 
$M$ from $P$ {\it once the modified location has been reached}. This transformation is depicted in Figure~\ref{fig:from-wm-to-labels}.    Let us consider  a mutant $M$ differing from $P$ 
only at location $\loc$. We consider three cases, depending on  the modification introduced by $op$:  
%
\begin{itemize}
\item {\tt lhs := expr}   becomes   {\tt lhs := expr'}:  we add  label  $\lab \mydef\  (\loc, expr \neq expr')$. We must prove that $\dt \covers_{P} M$ iff $\dt \covers_{P} \lab$. 
 Note that   $\dt \covers_P \lab$ iff $\dt$  reaches $\loc$ with an internal state such that 
{\tt expr} and  {\tt expr'} evaluate to  different values. This is equivalent to say that  $P(\dt)$ and $M(\dt)$ are in different internal states right after $\loc$, which corresponds 
 by definition to   $\dt \covers_P M$.

\item {\tt if (cond) then}  becomes {\tt if (cond') then}: we add  label $\lab \mydef\ (\loc, cond \oplus cond')$, where $\oplus$ is the xor-operator. 
We follow  the same line of reasoning as in the previous case.  The $\oplus$ operator ensures that $P(\dt)$  and $M(\dt)$ will not follow the same branching condition.

\item {\tt lhs := expr}   becomes   {\tt lhs' := expr}:   we add label  $\lab \mydef\ (\loc, \alpha(lhs) \neq \alpha(lhs') \wedge ( lhs \neq expr  \vee  lhs' \neq expr ))$, where 
 $\alpha(x)$ denotes the memory location ($\approx$ address) of $x$, not its value. For example, in C the memory location is given by the {\tt \&} operator.     This case requires a little bit more explanation. In order to observe a difference between $\Pg(\dt)$ and  $M(\dt)$ right after the mutated location,  
we need first that  {\tt lhs'} and {\tt lhs} refer to different memory locations (which is not always obvious in the case of aliasing expressions).  
Moreover, there are only two ways of noticing a difference: either the old value of {\tt lhs} differs from {\tt expr}, then {\tt lhs} will evaluate to different values in $P(\dt)$ 
(equals to {\tt expr})    
and in  $M(\dt)$ (remains unchanged) just after the mutation, or the symmetric counterpart for {\tt lhs'}. This is exactly what $\lab$ encodes.   

\end{itemize}

\noindent By applying this technique for every mutant created by all considered mutation operators, we obtain the desired labelling function. 
\end{proof}

\begin{figure}[htbp]
\begin{center}
 \includegraphics[width=\columnwidth]{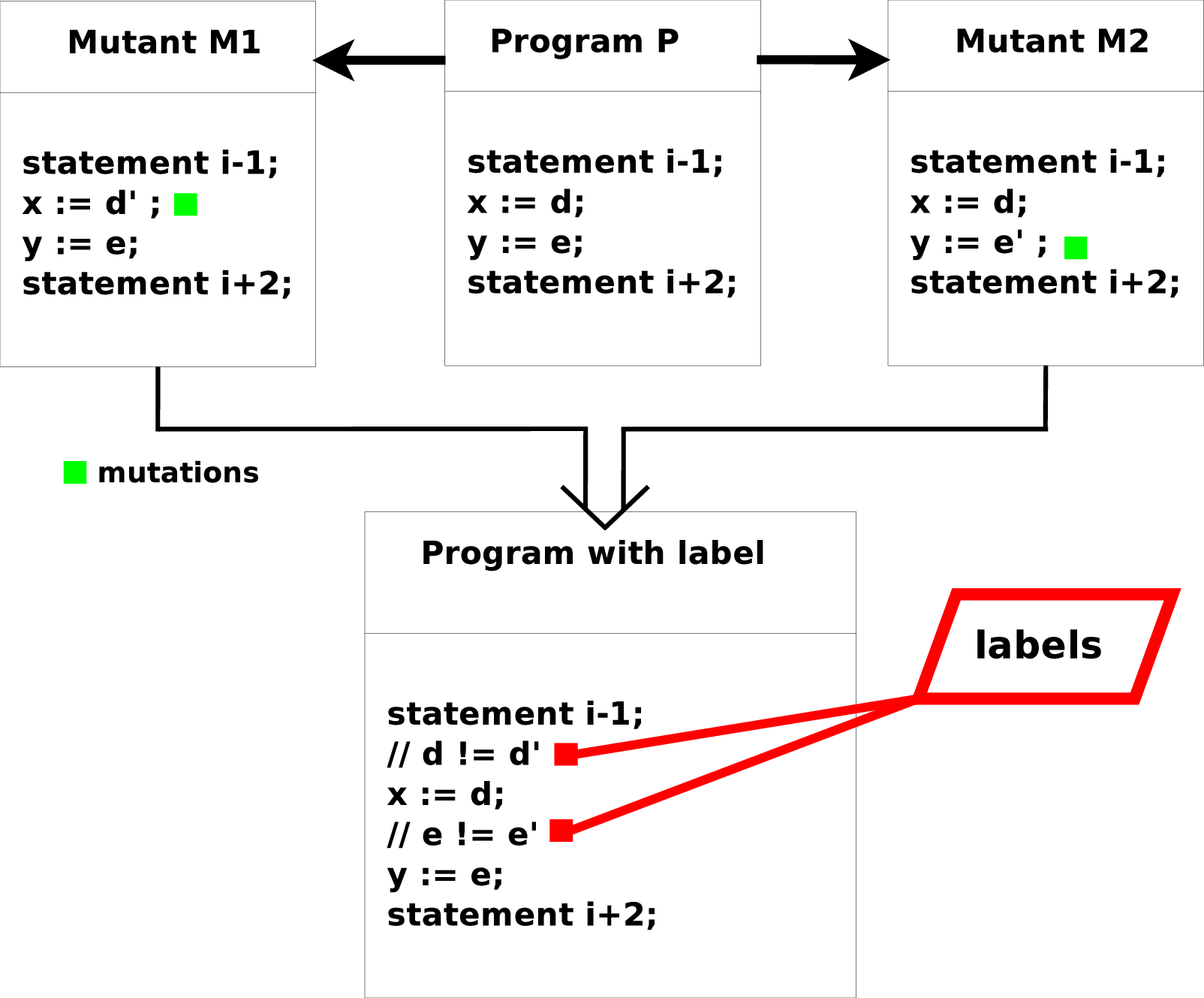}
\end{center}
\caption{Simulating weak mutants with labels}\label{fig:from-wm-to-labels}
\end{figure}

The subset of mutations we have been considering so far is limited to  
 (1) atomic mutations and (2) side-effect free operators. 
The first restriction is not a major issue as atomic mutations have been proved  to be almost as powerful as high-order mutations \cite{Offutt-92}. 
The second restriction has two sides: (2.a) it forbids mutation operators {\it introducing} side-effects, for example mapping {\tt x} to {\tt x++}, and (2.b) 
it forbids to mutate a side-effect prone expression.  
Again, restriction (2.a) is not severe: on side-effect free programs, the side-effect free fragment of  \WM\ encompasses the ABS, ROR, AOR, COR  and UOI operators \cite{AO-08}, which have been experimentally shown  
  mostly equivalent to much larger sets of operators~\cite{ORZ-93,WM-95}. 
It is left as an open question to quantify more precisely what is lost with restriction (2.b). 
Anyway, the previous points show that while side-effect free \WM\ is less powerful than full \WM, it is still a substantial subset.


%
%

\myparagraph{Other criteria.} 
Several other testing criteria commonly found in the literature can be emulated by labels. 
We focus here on Input Domain Coverage 
and Run-Time Error Coverage. 
\begin{itemize}
\item 
Input Domain Coverage:  assuming a partition  of the input domain $D$ of $\Pg$ given as disjoint predicates $\pred_1,\ldots,\pred_k$, 
this criterion consists in considering one $\dt_i$ for  each $\pred_i$.
 The corresponding labelling function 
adds all labels of the form  $(\loc_0,\pred_j)$, where $\loc_0$ is the entry point of $\Pg$. 
%
  The approach is independent of the way the partition is obtained, covering both interface-based  and functionality-based partitions \cite{AO-08}.

\item 
Run-Time Error Coverage: test objectives corresponding to  {\it run-time errors}  such as those implicitly searched for in active testing or assertion-based testing~\cite{CKGJ-10,GLM-08b,KA-96}   can be easily captured by labels. 
 These objectives include division by zero, out-of-bound array accesses or null-pointer 
dereference. Typically, any error-prone instruction at location $\loc$ with a precondition $\pred_{\text{safe}}$ will be tagged 
by a label $(\loc, \neg \pred_{\text{safe}})$.  
\end{itemize}




\myparagraph{Limits.} The following criteria cannot be emulated through labels, at least with simple encoding:  weak mutations with operators involving side-effects,  
criteria imposing constraints on paths rather than constraints on program locations (k-path coverage, data-flow criteria) and criteria relating  different paths (strong mutations, MCDC). 
It is left as future work to study whether these limitations are strict or not.

\section{Automating {LC}: a first attempt} \label{sec:naive}

Given an annotated program $\langle \Pg, \Lab  \rangle$, 
 we seek automatic methods: (1) to compute the \LC\ score of a given test suite $\TS$, and (2) to derive a test set achieving high \LC-coverage. 
We propose first  a black-box approach, reusing standard automatic testing tools through a {\it direct instrumentation} of $\Pg$. 
This technique underlies previous works aiming at extending DSE  coverage abilities~\cite{CKGJ-10,GLM-08b,KA-96,PMK-10}. 
While it allows for cheap \LC\ score computation, it is far from efficient for ATG, mainly   because of an exponential blow-up 
of the path space of the program.


     \subsection{Direct instrumentation}
  

The {\it direct instrumentation} $\Pg'$ for $\langle \Pg, \Lab \rangle$ 
consists in inserting for each label $\lab \mydef\ (\loc,\pred) \in \Lab$ a new branching instruction $I$: 
 {\tt if ($\pred$) \{\};}  such that    
%
all instructions leading to $\loc$ in $\Pg$  are connected   to $I$ in $\Pg'$, and $I$ is connected to $\loc$.  
The transformation is depicted in Figure~\ref{fig:typeI}. When different labels are attached to the same location, the new instructions are chained together in a sequence 
 ultimately leading to $\loc$.

\begin{figure}[htbp]

\begin{center} \small \tt 
\setlength{\columnseprule}{0cm}
\begin{tabular}{|c|c|c|}
\cline{1-1}\cline{3-3}

  && \\

\begin{lstlisting}
statement_1;
//label p 
statement_2; 
\end{lstlisting}
&$\longrightarrow$&
\begin{lstlisting}
statement_1;
if(p){};
statement_2;
\end{lstlisting}\\

  && \\ 

\cline{1-1}\cline{3-3}
\end{tabular}
\end{center}

\begin{center}
 \includegraphics[width = 0.8\columnwidth]{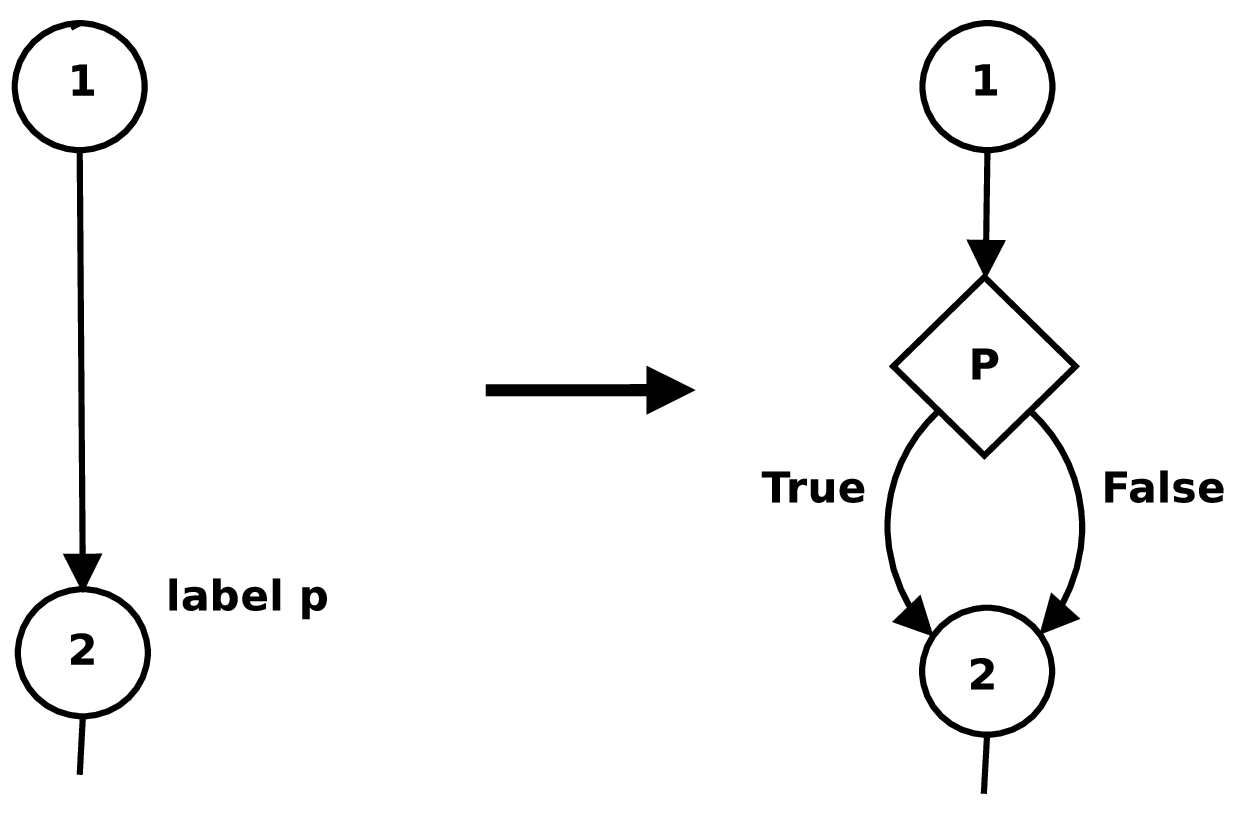}
\end{center}

\caption{Direct instrumentation $\Pg'$ }\label{fig:typeI}
\end{figure}

The direct instrumentation is sound with respect to \LC\  in the following sense. Let  us denote by {\bf NTD}  
the set of test objectives over $\Pg'$ requiring to cover all {\bf N}ew {\tt {\bf T}hen}-{\bf D}ecision introduced by the instrumentation. The following result holds.

          \begin{theorem}[Soundness]\label{thm:direct-sound} Given an annotated program $\langle \Pg, \Lab \rangle$, its instrumented version $\Pg'$ and a test suite $\TS$, we have:  
            $\TS \fullcovers_{\Pg} \LC$ iff  $\TS \fullcovers_{\Pg'} \text{\bf NTD}$. 
            \end{theorem}

This is interesting for both \LC\ score computation and ATG. 
Any ATG tool run on $\Pg'$ will  produce a test suite $\TS$ covering \LC\ for $\Pg$  as soon as $\TS$ covers all branches of interest in $\Pg'$.   
Concerning  score computation, a slightly modified version of the direct instrumentation, updating coverage information  in the new {\tt then}-branches,  
allows to compute \LC\ score efficiently. 

          \begin{theorem}\label{thm:coverage}   Given an annotated program $\langle \Pg, \Lab  \rangle$, its instrumented version $\Pg'$ and a test suite $\TS$, 
 then the \LC\ score of $\TS$ can be computed in time bounded by $|\TS| \cdot maxtime(\{\Pg'(\dt) | \dt \in \TS \})$.    
            \end{theorem}

 Interestingly, computing \LC\ score can be done independently from $\card{\Lab}$. Regarding coverage score computation, \LC\ is much closer to  \DC\ (each test $\dt$ is executed only once) 
than it is to \WM\ (each test $\dt$ is executed {\it once per mutant}).  While efficient mutation score computation is a difficult issue in mutation testing, 
 Theorem~\ref{thm:coverage} together with Theorem~\ref{thm:wm-lc}  show that the side-effect free subset of  \WM\ supports efficient mutation score computation.  
%
%

\subsection{Drawbacks} \label{sec:drawbacks}

So far, the direct instrumentation seems to perfectly suit our needs. 
Unfortunately, it is significantly inefficient for ATG. There are two main reasons for that.  
\begin{itemize}
\item $\Pg'$ is too complex:  it exhibits much more behaviours than $\Pg$, most of them  being unduly complex  for covering the labels we are targeting.   
\item DSE  will naturally produce a $\TS$ covering several times the same labels, which is useless since each label needs to be covered only once. 

\end{itemize}

\noindent We formalize the first point hereafter. We consider two dimensions in which  $\Pg'$ is ``too complex'': the size of the search space,  denoted $\card{\Paths(\Pg')}$, 
and the shape  of paths in $\Paths(\Pg')$. 
Let us call {\it label constraints} all additional branches $\pred$ and $\neg \pred$ introduced in $\Pg'$ compared to $\Pg$, and let us 
denote by $m$ the maximal number of labels {\it per} location in $\Pg$.  
A single path $\sigma \in \Pg$  may correspond to up to $2^{m\cdot\card{\sigma}}$ paths in $\Pg'$, since  each label of $\Pg$ creates a branching in $\Pg'$ and at most $m$ such branchings can be found at each 
step of $\sigma$.  Note also  that the paths $\sigma' \in  \Pg'$ corresponding to $\sigma \in \Pg$  have length bounded by $m\cdot\card{\sigma}$. Therefore they can pass through up to  $m\cdot\card{\sigma}$ label constraints, 
   while (by definition)  $\sigma$  does not pass through any label constraint. 
%
%
%
%
%
Theorem \ref{thm:non-tight} summarises these results.

\begin{theorem}[Non-tightness] \label{thm:non-tight}
Given an annotated program $\langle \Pg, \Lab  \rangle$ and its instrumented version $\Pg'$, let us assume that $\Paths(\Pg)$ is bounded,  
 that $k$ represents the maximal 
length of paths in  $\Paths(\Pg)$ and that $m$ is the maximal number of labels {\it per} location in $\Pg$.   
Then $\Paths(\Pg')$ is more complex than  $\Paths(\Pg)$ in the following sense: 
\begin{itemize}

\item $\card{\Paths(\Pg')}$  can be exponentially larger than $\card{\Paths(\Pg)}$   
by a factor $2^{m \cdot k}$;

\item any $\sigma' \in \Paths(\Pg') $ may carry up to $m \cdot k$ (positive or negative) label constraints.   

\end{itemize}

\end{theorem}


Both aspects are problematic for a symbolic exploration of the search space: more paths means either more requests to a theorem prover (DSE) or a larger formula (BMC), and 
more constrained paths means more expensive requests.

\section{Efficient ATG for LC }  \label{sec:smart}

We describe in this section two main ingredients in order to obtain efficient ATG for \LC: 
(1) {\it a tight instrumentation}  avoiding all  drawbacks of the direct instrumentation, 
and (2)  a strong coupling of label coverage and DSE through {\it iterative label deletion}. 

\subsection{Tight instrumentation } \label{sec:tight-instrum}

Given a label $\lab \mydef\ (\loc,\pred)$, 
the key insights behind the tight instrumentation are the following: 
\begin{itemize}
  \item   label constraint $\pred$ is useful only for covering $\lab$, and should not  be propagated beyond that point;  

\item    label constraint $\neg \pred$ is pointless w.r.t.~covering $\lab$,  and should not be enforced in any way.  
\end{itemize}

Keeping  these lines in mind, 
the instrumentation works as depicted in Figure~\ref{fig:tight}: for each label $(\loc,\pred)$, we introduce a new instruction    
\ {\tt if (nondet) \{assert($\pred$); exit\};} 
where {\tt assert($\pred$)} requires $\pred$ to be verified, {\tt exit} forces the execution to stop and {\tt nondet} is a non-deterministic choice. 
In the resulting instrumented program $\Pg^{\star}$ (Figure~\ref{fig:tight}, right column), when an execution  reaches $\loc$, it gives rise to two execution paths: the first one tries to cover the label by 
asserting $\pred$ and {\it stops right there},   
the second  one simply follows its execution {\it as it would do in} $\Pg$, neither $\pred$ nor $\neg \pred$ being enforced. 

\begin{figure}[htbp]
\begin{center} \small  \tt 
\setlength{\columnseprule}{0cm}
\begin{tabular}{|c|c|c|}
\cline{1-1}\cline{3-3}
  && \\ 
  
\begin{lstlisting}
statement_1;
// label p
statement_2; 
\end{lstlisting}
&$\longrightarrow$&
\begin{lstlisting}
statement_1;
if(nondet){
    assert(p);
    exit(0);
};
statement_2;
\end{lstlisting}\\
  && \\ 
   
\cline{1-1}\cline{3-3}
\end{tabular}
\end{center}

\begin{center}
 \includegraphics[width = 0.8\columnwidth]{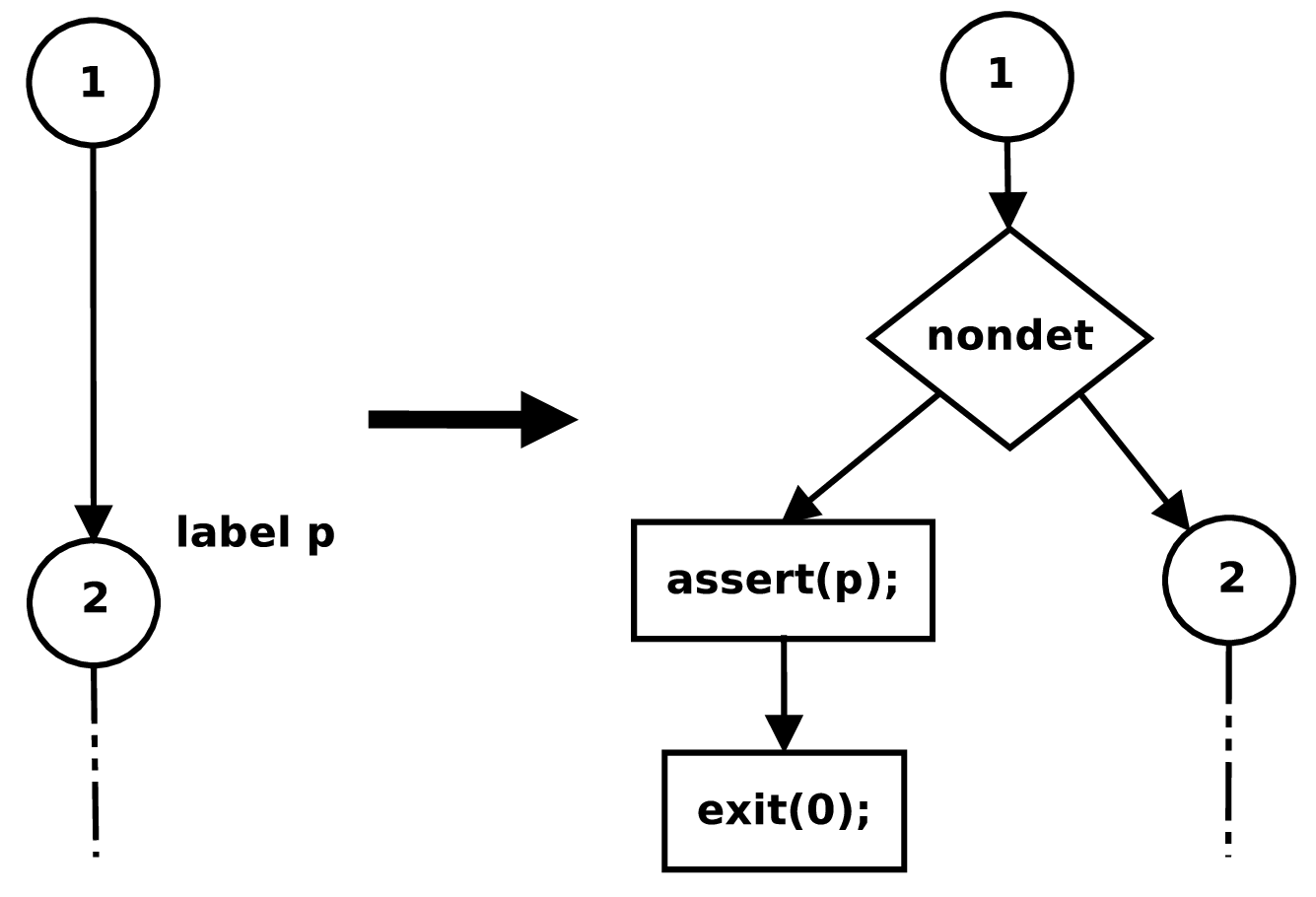} 
\end{center}
\caption{Tight instrumentation $\Pg^{\star}$} \label{fig:tight}
\end{figure}

The tight instrumentation $\Pg^{\star}$ is   sound w.r.t.~\LC. Let us denote by {\bf NA}  
the test objective over $\Pg^{\star}$ requiring to cover all {\bf N}ew {\tt {\bf A}ssert} introduced by the instrumentation (with condition evaluating to true). The following result holds.

          \begin{theorem}[Soundness]\label{thm:tight-sound} Given an annotated program $\langle \Pg, \Lab \rangle$, its tight instrumentation $\Pg^{\star}$ and a test suite $\TS$, we have:   
              $\TS \fullcovers_{\Pg} \LC$ iff  $\TS \fullcovers_{\Pg^{\star}}  \text{\bf NA} $. 
            \end{theorem}

Interestingly, the tight instrumentation does not show any of the issues reported in Theorem \ref{thm:non-tight}. The underlying reasons have been sketched at the beginning 
of Section \ref{sec:tight-instrum}  and are depicted in Figure~\ref{fig:comparison}. 
A single execution path in $\Pg$ going through $n$ labels 
 can give birth up to $2^{n}$ paths in $\Pg'$ (left column), while it can create only 
  $n+1$ paths in  $\Pg^{\star}$ (right column). 
Moreover, each path in $\Pg^{\star}$ can go through  at most one single positive label constraint, while a path $\sigma'$ in $\Pg'$ can carry up to 
$\card{\sigma'}$  (positive or negative) label constraints.  
These results are summarized in Theorem~\ref{thm:tight-tight}. 

\begin{figure}[htbp]
\begin{center}
  \includegraphics[width = \columnwidth]{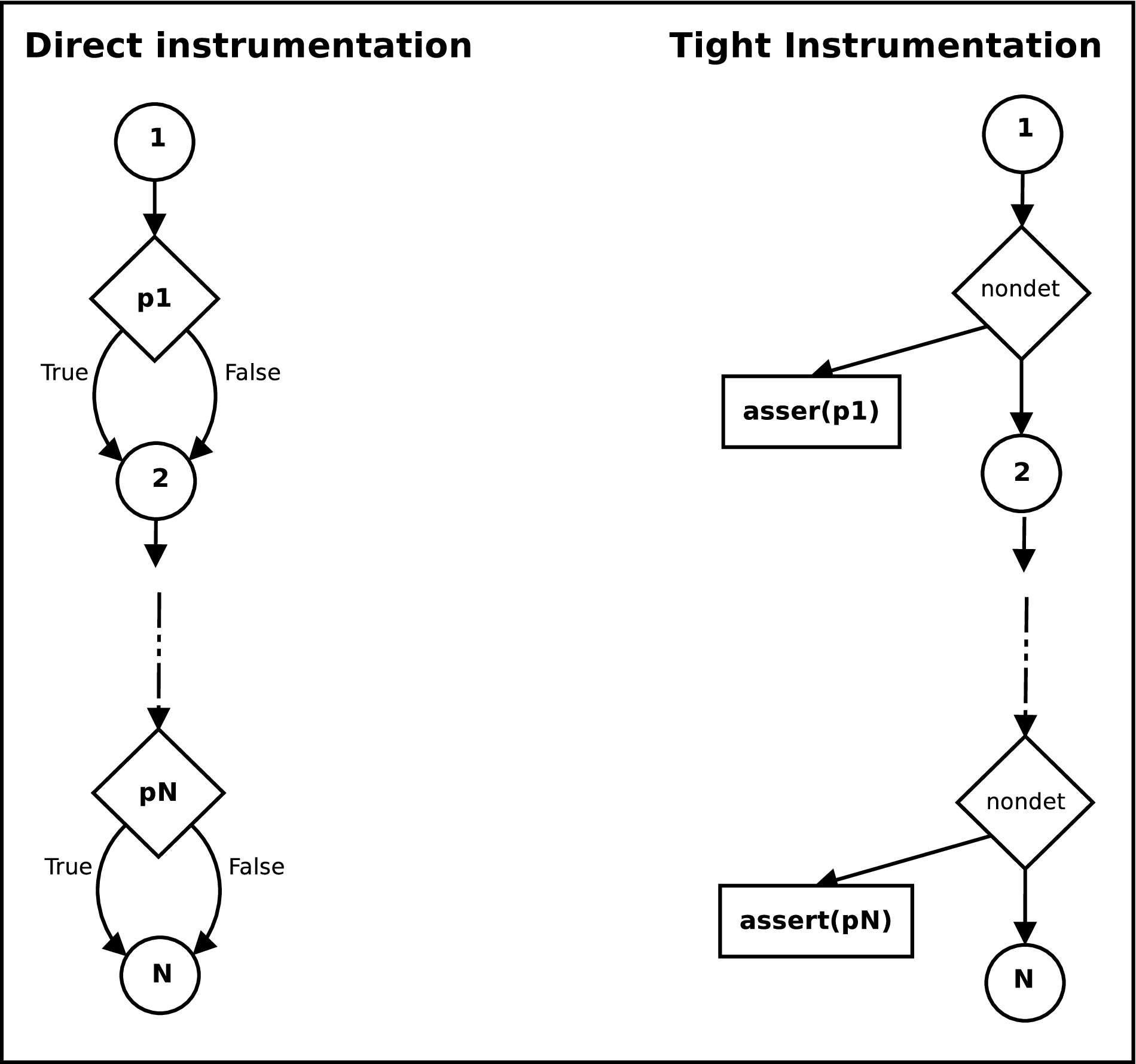}
\end{center}
\caption{Direct vs.~tight instrumentation}\label{fig:comparison}
\end{figure}

\begin{theorem}[Tightness]\label{thm:tight-tight}
Given an annotated program $\langle \Pg, \Lab  \rangle$ and its instrumented version $\Pg^{\star}$,  let us assume that $\Paths(\Pg)$ is bounded, that $k$ represents the maximal 
length of paths in  $\Paths(\Pg)$ and that $m$ is the maximal number of labels {\it per} location in $\Pg$.   
Then $\Pg^{\star}$ is tight in the following sense:  
\begin{itemize}

\item $\card{\Paths(\Pg^{\star})}$  is linear in $\card{\Paths(\Pg)}$ and $m\cdot k$;

\item 
any $\sigma \in \Paths(\Pg^{\star}) $ carries at most one label-constraint. 
\end{itemize}
\end{theorem}

\begin{proof}
The main reasons behind this result directly follow from the tight instrumentation, and have already been exposed just before Theorem~\ref{thm:tight-tight}. 
We can be more precise:    
$\card{\Paths(\Pg^{\star})}$    is bounded by $(m\cdot k +1)\cdot \card{\Paths(\Pg)}$.  
\end{proof}

Theorem~\ref{thm:tight-tight} implies that  
any path-based method conducted over $\Pg^{\star}$ will have a much easier task 
than  over $\Pg'$, since  $\Pg^{\star}$ contains exponentially less paths  and those paths are simpler. This is independent of the underlying verification technique 
 as long as it enumerates paths in some way, such as DSE or BMC.

        \subsection{Iterative label deletion} \label{sec:deletion}

We focus now on the last issue pointed out in Section~\ref{sec:drawbacks}. 
Besides the dramatic complexity of the search space induced by $\Pg'$, 
which is settled by $\Pg^{\star}$, the remaining problem is that 
a DSE procedure launched on $\Pg^{\star}$ will try to cover all paths 
from $\Pg^{\star}$, while we are only interested in covering branches corresponding 
 to labels.  Especially, a standard DSE may try to cover many path prefixes 
ending in an already-covered {\tt assert($\pred$)}.  Whether they fail or not, these computations will be redundant 
since the \LC-soundness result of Theorem~\ref{thm:tight-sound} requires  that each new {\tt assert} be covered only once.

\medskip 

{\it Iterative label deletion} (IDL) consists in (conceptually) erasing a label constraint   as soon as it is covered, so that it  will not affect 
the subsequent path search.  
IDL requires to modify  SE/DSE  in the following way:  each label $\lab$ is equipped with a boolean variable $b_{\lab}$ set to true iff $\lab$  
has already been covered during  path exploration,  and attempts to symbolically execute paths leading to $\lab$ continue as long as $b_{\lab}$ is false. 

We present  DSE with IDL over annotated programs in Algorithm~\ref{algo:se-idl}, where modifications w.r.t.~standard SE/DSE are pointed out by $(\star)$ marks.   
We assume that  $\Paths(\langle \Pg, \Lab \rangle)$ is constructed in the following way: at each step, a run encountering a label $\lab \mydef\ (\loc,\pred)$ can either  
 choose to go through $\lab$ (enforcing $\pred$) and continue,  or bypass $\lab$ (no constraint) and continue. 
%
 The adaptation to $\Pg^{\star}$ is described after.

\begin{algorithm}
\KwIn{an annotated program $\langle \Pg, \Lab \rangle$ with finite set of paths  $\Paths(\Pg)$ }
\KwOut{$\TS$, a set of pairs $(\dt,\sigma)$ such that $\Pg(\dt) \covers_{\Pg} \sigma$ }

 $\TS$ := $\emptyset$\; 
 $S_{paths}$ :=   $\Paths(\langle \Pg, \Lab \rangle)$\;

  \While{$S_{paths} \neq \emptyset$}{
     choose $\sigma \in S_{paths}$; $S_{paths}$ := $S_{paths} \backslash \{\sigma$\}  \;

     compute $\pred_{\sigma}$\; 


 \Switch{solve($\pred_{\sigma}$)}{

 \uCase{sat($\dt$): }{  

   $\TS$ := $\TS \cup  \{(\dt,\sigma)\}$\; 

   ($\star$) for all $\lab$ covered by $\sigma$, do: $b_{\lab}$ := $1$ \;    

   ($\star$) remove from  $S_{paths}$  all $\sigma'$ going through \\ \quad  a   label $\lab$  s.t.~$b_{\lab} = 1$ \;    

 }

 \lCase{unsat: }{ skip }

  }
}

\Return{$\TS$}\;

\caption{Symbolic Execution with IDL}\label{algo:se-idl}
\end{algorithm}

For  integration in a realistic SE/DSE setting with dynamic exploration of path space,  
we distinguish  two flavors of IDL:   
\begin{description}
\item[{\sc idl-1}] a label is marked as covered   only when  it belongs to a  path prefix  being successfully solved. 
 This is a purely symbolic approach.      

\item[{\sc idl-2}] a label is also marked   when it is covered by a  concrete execution, taking   
 advantage of dynamic runs   
to delete several labels at once.  
%
%
%
\end{description}

\myparagraph{Combining {\sc idl} with tight instrumentation.} 
Both variants of {\sc idl} can be combined with tight instrumentation either in a dedicated manner  or in a black-box setting.  Since dedicated implementations are straightforward,  
 we focus hereafter on  black-box implementations.  
An instrumentation enforcing {\sc idl-1} over $\Pg^{\star}$ is depicted in Figure~\ref{fig:typeII}. 
We follow the idea of adding extra boolean variables  for coverage, denoted {\tt b\_l} where {\tt l} is a label identifier. However it is mandatory that the coverage information be global to the whole path search process  
and not bound to a single execution. 
It can be  achieved by putting the coverage information in an external file, accessed and modified through operations  {\tt read(b\_l)}  and 
{\tt set\_covered(b\_l)}. 

\begin{figure}[htbp]
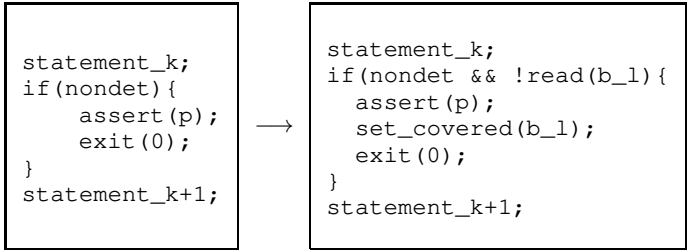
 
\small \tt 
\begin{center}
\setlength{\columnseprule}{0cm}
\begin{tabular}{|c|c|c|}
\cline{1-1}\cline{3-3}
 
  && \\ 
\begin{lstlisting}
statement_k;
if(nondet){
    assert(p);
    exit(0);
}
statement_k+1;
\end{lstlisting}
&$\longrightarrow$&
\begin{lstlisting}
statement_k;
if(nondet && !read(b_l){
  assert(p); 
  set_covered(b_l);
  exit(0);
}
statement_k+1;
\end{lstlisting}\\
  && \\ 
  
\cline{1-1}\cline{3-3}
\end{tabular}
\end{center}
\caption{{\sc idl-1} variant of tight instrumentation $\Pg^{\star}$}\label{fig:typeII}
\end{figure}

Enforcing {\sc idl-2} in a black-box setting requires a fine-grained control over the DSE procedure. We need to be able to 
query the DSE engine for the next generated test data. The procedure reuses  the {\sc idl-1} 
approach for path space exploration, but  each new generated test data  is also run 
on the {\it direct instrumentation} $\Pg'$. Then  all covered labels are marked in the  coverage file of  {\sc idl-1}  before 
the next test data is searched for. 
 The technique is depicted in Figure~\ref{fig:update2}, where TD stands for ``test data''.

\begin{figure}[htbp]
\begin{center}
 \includegraphics[width=\columnwidth]{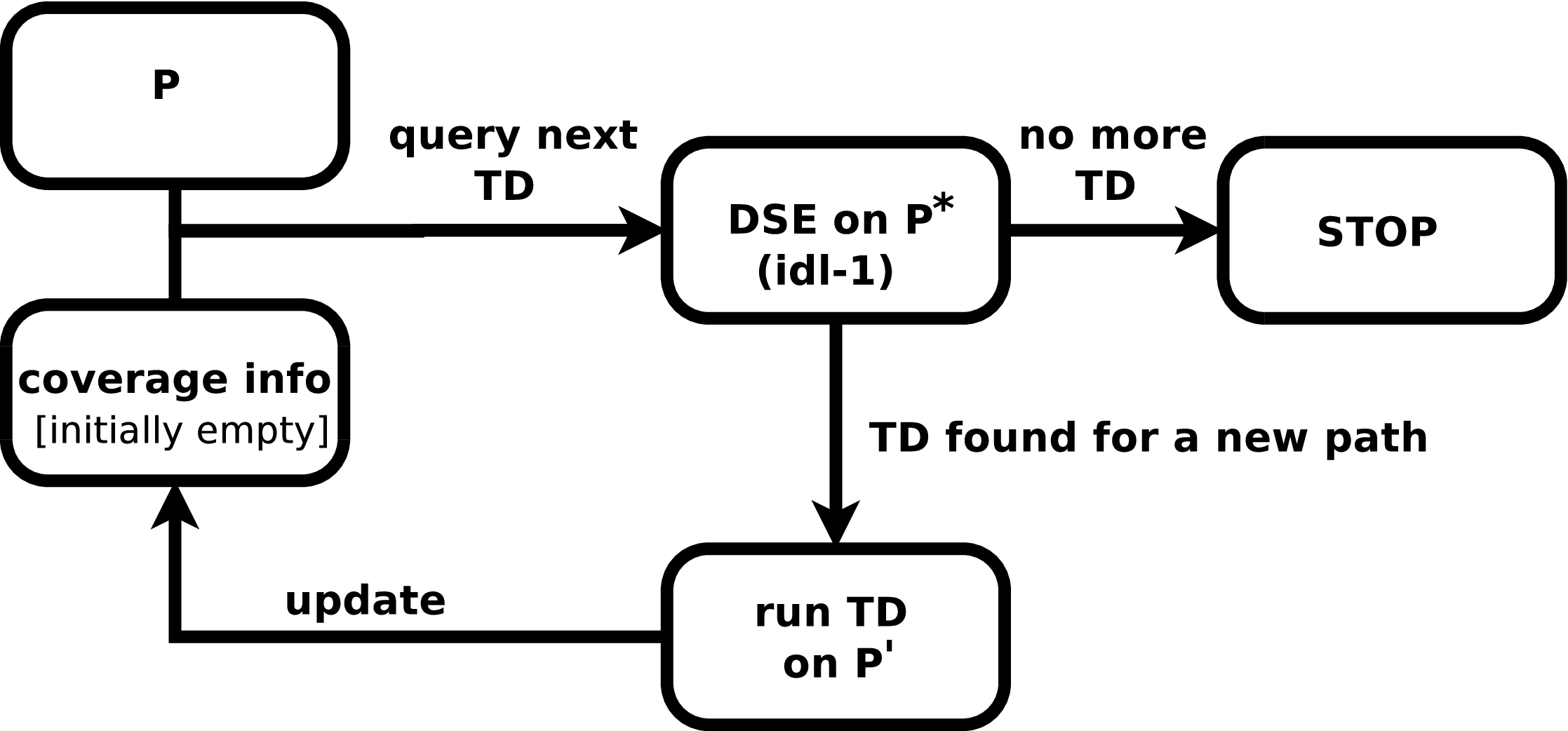}
\end{center}
\caption{{\sc idl-2} variant for DSE }\label{fig:update2}
\end{figure}

We denote by DSE$^{\star}$ the DSE procedure enhanced with {\sc idl-1} or {\sc idl-2}, and we consider 
only deterministic  and sound DSE techniques. The following result holds. 

\begin{theorem} 
Given an annotated program $\langle \Pg, \Lab \rangle$ and its tight instrumentation $\Pg^{\star}$, then
DSE$^{\star}$($\Pg^{\star}$) covers as many labels as DSE($\Pg^{\star}$) does. 
\end{theorem}

\begin{proof}
The result comes from three facts. First, a label is discarded iff it is covered by an already generated test data $\dt \in \TS$. 
Second, labels act only as ``observers'' in $\langle \Pg, \Lab \rangle$:  
they do not impact the execution, so they cannot enable or prevent the coverage of a particular test objective. It implies that paths from $\Pg$ are sufficient for reaching coverable labels ({\sc Fact 2}).   
Third, by construction, the set of paths from   
  $\Pg^{\star}$ contains all paths from $\Pg$ (those without any label constraint) plus additional paths with label constraints ({\sc Fact 3}). Therefore, deleting a label constraint in DSE$^{\star}$($\Pg^{\star}$) cannot discard any of the original path from $\Pg$. Using {\sc Fact 2}, we deduce that label deletion  cannot make uncoverable an otherwise-coverable label.    Note that   
the proof (and the theorem) does not hold for the direct  instrumentation $\Pg'$ since    {\sc Fact 3} is false in that case.   
\end{proof}

\section{Implementation \& Experiments}  \label{sec:experiments}

\subsection{Implementation}

We have implemented tight instrumentation and iterative label deletion 
inside  \pathcrawler~\cite{WMM-04a}. The tool follows a standard  DSE approach and  targets safety-critical C programs, with  a strong focus  on relative completeness guarantees.   
For example, the underlying constraint solver deals precisely with modular arithmetic, bitwise operations,  floats and multi-level pointer dereferences.   
The DSE engine relies on a simple DFS path search heuristics. On the other hand, the tool is highly optimised for programs 
with many infeasible paths. Aside from early detection of infeasible paths which is ensured by the ``dynamic execution-driven'' nature of DSE,   
  optimisations 
include several levels of incremental solving for cheap detection of infeasible paths. 

Our implementation  follows  the description of Section~\ref{sec:smart}. 
We adopt a grey-box approach rather than a full black-box approach because \pathcrawler\ does not offer yet the required API for {\sc idl-2} 
and does not support non-deterministic choice. 
We add a {\tt pathcrawler\_label(bool)} instruction, those native treatment implements tight instrumentation and {\sc idl-2}.  
The current search heuristics is mostly depth-first, but labels are handled as soon as possible.  











\subsection{Experiments}

Preliminary experiments have been conducted in order to test the following properties:  
%
$(i)$ the relative gain of our two optimisations w.r.t.~direct instrumentation, 
$(ii)$ the overhead of leveraging DSE  to \LC.   
%
%
Evaluating the practical feasibility of label-based DSE  over large programs 
or its  bug-finding power 
are left as future work.  
Note  that the bug-finding power of the coverage criteria emulated through labels in Section \ref{sec:labels}  
has already been extensively studied in the literature.

\myparagraph{Protocol.} We consider a few standard benchmark programs  taken from related works~\cite{CKGJ-10,PMK-10,PM-10}, together with   
three types of labels  simulating standard coverage criteria of increasing difficulty: \CC, \MCC\ and \WM. 
For \WM, our labels  mimic  mutations typically introduced    
by   MuJava~\cite{MOK-06}  for operators AOIU, AOR, COR and  ROR \cite{AO-08}.  
We compare the following algorithms: DSE($\Pg$) denotes the standard DSE on the standard program (witness),  
 DSE($\Pg'$) denotes  standard DSE on  direct instrumentation,  DSE($\Pg^{\star}$) denotes  standard DSE on  tight instrumentation and 
DSE$^{\star}$($\Pg^{\star}$) denotes DSE with iterative label deletion run on  tight instrumentation.   
Experiments are performed on a standard laptop (Intel Core2 Duo 2.40GHz, 4GB of RAM). Time out for solver is set to 1 min.

We record the following information: number of paths explored by the search, computation time and achieved coverage. 
The number of paths is a good measure for comparing the complexity of the different search spaces, and therefore to assess both the ``cost'' 
of leveraging DSE to labels and the benefits of our optimisations. 
%
%
Coverage score together with computation time  indicate how practical label-based DSE is.


It must be highlighted that \pathcrawler\ does not stop until all feasible paths are explored. This strategy gives us a good estimation of the size of the path space, 
however in pratice it would be wiser to implement a label-based stopping criteria.  Hence, from a feasibility point of view, results reported here are too pessimistic.


\begin{table}[htb]
\begin{center}
{\scriptsize
\begin{tabular}{|c|l|c|c|c|c|}
\cline{3-6}
\multicolumn{2}{c|}{} & \bigstrut DSE($P$) & DSE($P'$) & DSE($P^{\star}$) &  DSE$^{\star}$($P^{\star}$)\\
\multicolumn{2}{c|}{}                & \bigstrut (witness)  &  &   & \\

\hline
Trityp-cc        & \#p &    35    &  183  & 83   & 46   \\

  50 loc         & time  &  1.3s    & 1.6s   &  2s  &  4.5s  \\

  24 $\lab$     &  cover  &        &  24/24  & 24/24    & 24/24   \\

\hline
Trityp-mcc        & \#p &    35    &  337  & 110   & 66   \\

  50 loc         & time  &  1.3s    & 1.9s   &  3s  &  2.1s  \\

  28 $\lab$     &  cover  &        &  28/28  & 28/28    & 28/28   \\

\hline  
Trityp-wm        & \#p &    35    &    x  & 506   & 48   \\

  50 loc         & time  &  1.3s    &  x  &  12s  &  5.1s  \\

  129 $\lab$     &  cover  &         &  x  & 120/129    & 120/129   \\

\hline
4Balls-wm            & \#p       & 7     & \bf 195   & 75   &  \bf 23  \\
  35 loc              & time     & 1.2s    & 1.9s   & 2.1s   & 2.1s   \\

  67 $\lab$         &  cover     &      & 56/67   &  56/67   &  56/67  \\

\hline
utf8-3-wm          &  \#p     & 134   & 1,379   &  626   &  313  \\
 108 loc           & time     & 1.4s   &  4.2s  &  4.3s   &  3.8s  \\

84 $\lab$          & cover    &      &  55/84  &  55/84  &  55/84  \\

\hline
utf8-5-wm          &  \#p     & 680   & \bf 11,111   &  3,239  & \bf 743  \\
 108 loc           & time     &  2s  &  40s  &  24s  &  8.1s  \\

84 $\lab$          & cover    &      &  82/84  &  82/84  &  82/84  \\

\hline
utf8-7-wm          &  \#p     & 3,069   & \bf 81,133   & 14,676    & \bf 3,265  \\
 108 loc           & time     &  5.8s   & \bf 576s     &  110s     & \bf 35s   \\

84 $\lab$          & cover    &      &  82/84  &  82/84  &  82/84  \\

\hline
 Tcas-cc           &  \#p     & 2,787     & 3,508   & 3,508   &  2,815  \\
124  loc           & time     &  2.9s    &  3.6s   &  5s  &  3.4s  \\

10 $\lab$          & cover    &      & 10/10   & 10/10   & 10/10   \\

\hline
Tcas-mcc           &  \#p     & 2,787    & 3,988   & 3,988   &  3,059  \\
124  loc           & time     &  2.9s   & 4.2s    & 5.2s   & 3.9s   \\

12 $\lab$          & cover    &         & 11/12    & 11/12   &  11/12  \\

\hline
Tcas'-wm        &  \#p     & 4,420    & \bf 300,213   & 20,312   & \bf 6,014   \\
124  loc        & time     &  5.6s   & \bf 662s   & 120s   & \bf 27s   \\

111 $\lab$      & cover    &      &  101/111  & 101/111   & 101/111   \\


\hline
Replace-wm          &  \#p     & 866    & \bf 87,498  & 6,420   & \bf 2,347   \\
  100 loc           & time     & 2s     & \bf  245s   &  64s  & \bf  14s \\

 79 $\lab$          & cover    &      & 70/79   & 70/79   & 70/79   \\

\hline
\end{tabular}
}


\end{center}

{\small x:  crash due to a bug in the underlying solver}

\caption{Experimental results for ATG}  \label{tab:results:atdg2}
\end{table}

\myparagraph{Results.} Results are summarized in Table~\ref{tab:results:atdg2}. 
We can observe the following facts. 
First, the number of explored paths is always  much greater in DSE($\Pg'$) than in DSE$^{\star}$($\Pg^{\star}$), with a factor between 4x and 50x on all examples but  {\tt Tcas-cc} and {\tt Tcas-mcc}, where the difference is less than 1.5x.  
Actually, labels in those two programs  lead mostly to infeasible paths, yielding no  path explosion.  
Things are more mitigated for computation time, certainly because of \pathcrawler\ optimisations. Yet, DSE($\Pg'$) and DSE$^{\star}$($\Pg^{\star}$) are at worst mostly equivalent on the smallest 
examples (except for  {\tt Trityp-cc}),  
  and DSE$^{\star}$($\Pg^{\star}$) can   be up to 25x faster on the most demanding programs.  
Second, as expected, DSE($\Pg^{\star}$) stands   between DSE($\Pg'$) and DSE$^{\star}$($\Pg^{\star}$) for the number of paths, and for computation time on the most demanding examples.  
Finally, compared with  DSE($\Pg$), DSE$^{\star}$($\Pg^{\star}$) yields at most a 3x growth of the search space and  an overhead between 1.1x and 7x for computation time.




\myparagraph{Conclusion.} These experiments  confirm our formal predictions: 
\begin{itemize}

\item our fully-optimised DSE performs significantly better on difficult programs  than the direct instrumentation, both in  terms of search space and computation time;

\item the overhead w.r.t.~standard DSE  turns out to be  always acceptable, even sometimes very low.   

\end{itemize}

These preliminary results  suggest that DSE  can be  efficiently leveraged to \LC\ coverage thanks to our optimisations. Yet, 
additional experiments on real-size programs are required to confirm that point.

Finally, we can observe that the very significant reduction of the path space does not always translate into an equivalent reduction of computation time. 
Optimisations of DSE tools certainly  play a role here.  
%
Indeed, several \pathcrawler\ 
optimisations 
may significantly accelerate usual DSE exploration for some programs.


\section{Beyond test data generation}  \label{sec:framework}


Section \ref{sec:labels} proves that \LC\ is a powerful coverage criterion, 
encompassing many standard criteria and a large subset of weak mutations. 
  Section~\ref{sec:smart} and Section~\ref{sec:experiments} demonstrate the feasibility of efficient ATG for   \LC, 
with a cost-effective integration in  DSE.  
We also sketched in Section~\ref{sec:naive} how to perform cheap \LC\ score computation.    
Everything put together, labels form the basis of a very powerful framework for automatic testing, handling many different criteria in a uniform fashion.  
We describe such a view in Figure~\ref{fig:framework}.

\begin{figure}[htbp]
\includegraphics[width = \columnwidth]{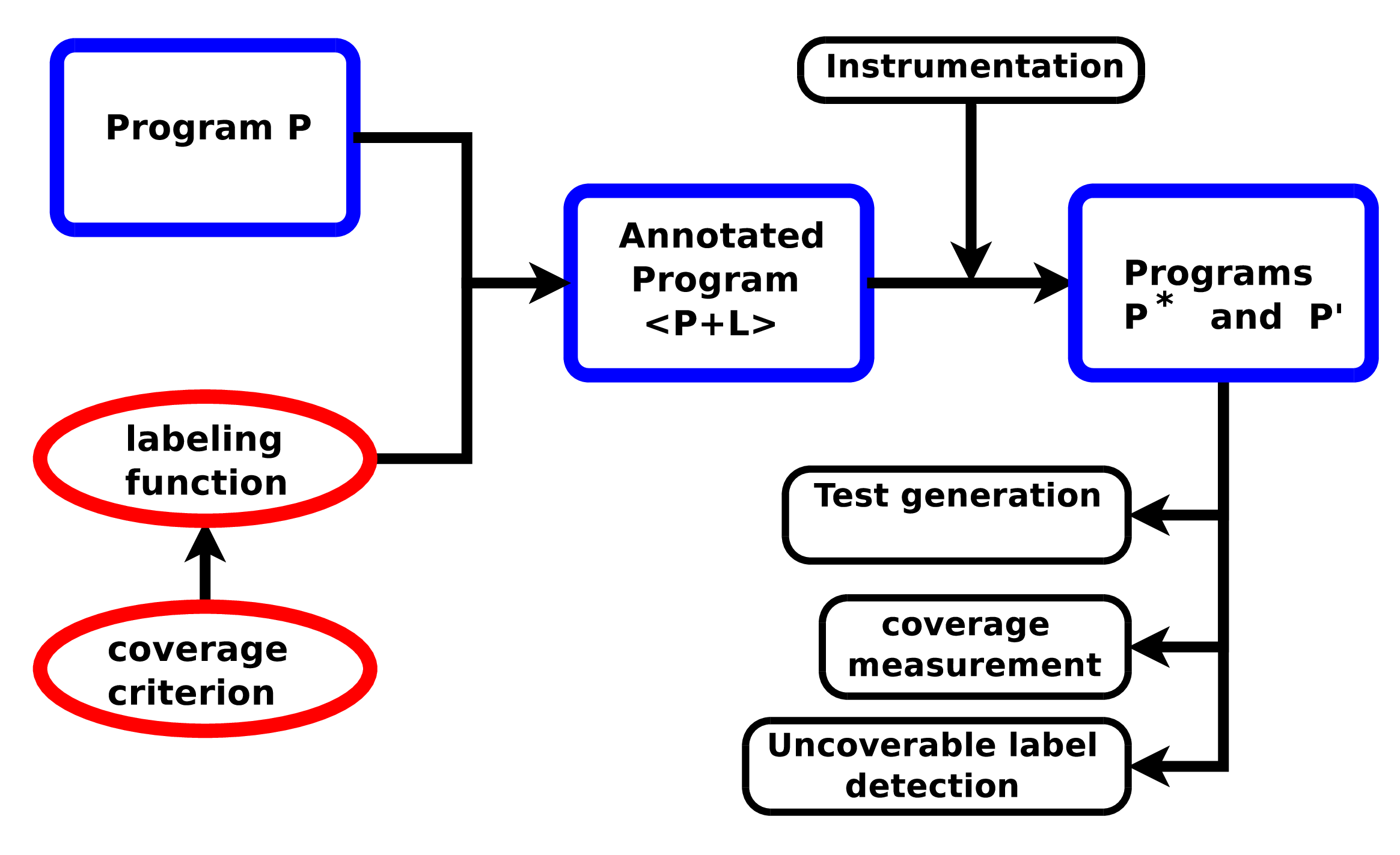}
\caption{\LC-coverage framework} \label{fig:framework} 
\end{figure}

Starting from a program $\Pg$ and a testing criterion $\C$, a predefined labelling function $\labfun_{\C}$  creates the $\C$-equivalent annotated program $\langle \Pg, \Lab \rangle$ (Theorem~\ref{thm:standard-lc} and Theorem~\ref{thm:wm-lc}). Then, we can perform automatic and efficient \LC\ score computation and \LC-based ATG through instrumentation (Theorem~\ref{thm:coverage} and Theorem~\ref{thm:tight-sound}). 
Finally, static analysis techniques  can be used on $\Pg^{\star}$ in order to detect uncoverable  
labels, i.e.~labels $\lab \mydef\ (\loc,\pred)$ for which there is no test data  $\dt$ such that $\dt \covers_{\Pg} \lab$. 
Static detection of uncoverable labels can help ATG tools by  avoiding wasting  time on infeasible objectives, as was observed in the case of mutation testing~\cite{JKS-12a}.   

\section{Related work}  \label{sec:related}

\myparagraph{Leveraging DSE to higher coverage criteria.} The need for enhancing DSE with better coverage criteria has already been pointed out in active testing (a.k.a~assertion-based testing) \cite{CKGJ-10,GLM-08b,KA-96} 
 and  in Mutation DSE \cite{PM-10,PM-11}.   
The present  work  generalizes these results and proposes ways of taming the potential blow-up, resulting in an effective support of advanced coverage criteria in DSE with only a  small overhead. 

Active testing targets  run-time errors by adding explicit branches into the program. It is similar to the Run-Time Error Coverage 
criterion presented in Section~\ref{sec:labels}. Labels are a more general approach. 
Interestingly, the direct instrumentation $\Pg'$ for this criterion is mostly equivalent to $\Pg^{\star}$ since additional branches can only trigger 
errors and stop the execution. Yet, active testing could benefit from the IDL optimisation. In that case only the {\sc idl-1} flavour makes sense since an execution cannot cover two 
 different run-time errors.  Finally, since most test objectives are (hopefully!) uncoverable for Run-Time Error Coverage, some approaches aim at combining DSE with static detection 
of uncoverable targets~\cite{CKGJ-10}. These techniques and heuristics can be reused for labels, and should be useful when many labels are uncoverable.

Mutation-based DSE \cite{PM-10,PM-11,PMK-10} is probably the work closest to ours. Following Offut {\it et al.} \cite{DO-91},  Papadakis {\it et al.} show that \WM\ can be reduced to 
branch coverage through the use of a variant of Mutant Schemata \cite{UOH-93}. 
This is pretty similar to  the direct encoding $\Pg'$ mentioned here. They propose essentially two variations of DSE for mutation testing: 
a black-box approach \cite{PM-10} based on a direct encoding similar to our DSE($\Pg'$) scheme,  
and a more ad hoc approach \cite{PM-10} preventing reuse of existing DSE tools but offering several optimisations.  
Papadakis {\it et al.} propose a variant of IDL, a dedicated search heuristic based on shortest paths \cite{PM-09} 
and an improvement of the direct encoding through the use of mutant identifiers  (following exactly Mutant Schemata).  
On the one hand, it ensures that a given path cannot go through several {\it different} mutants, on the other hand there is still an exponential blow-up 
of the search space in the worst case, and IDL cannot cover more than one mutant at once. 

We give a more generic view of the problem, identifying  labels and annotated programs  as  the key concept underlying the approach. 
We also clearly identify the limits and hypothesises of the method by defining the side-effect free fragment of \WM, proving soundness of direct instrumentation  and 
  providing a formalization of the path space ``complexification'' (Non-Tightness Theorem) induced by  direct instrumentation.   
Most important, we propose the tight instrumentation  which  completely prevents  complexification.  Finally, our optimisations  can be implemented in a pure  black-box setting   
 and we do not impose anything on the search heuristics, keeping room  for future improvements.

\myparagraph{Labels and optimized DSE.} The label-specific optimisations described here can be freely mixed  with other DSE optimisations.  
It is left as future work to explore which optimisations turn out to be the most effective for labels. As already stated, combining static discovery of uncoverable labels 
with DSE \cite{CKGJ-10} could be useful for often-uncoverable labels, such as those generated for Run-Time Error Coverage 
 or \MCC.   Another interesting idea is to adapt 
   DSE  search heuristics \cite{XTHS-09} by taking advantage of the dissimilarities between labels and branches, possibly getting inspiration from \cite{PM-09}.  

The IDL optimisation shows some similarities with  Look-Ahead pruning (LA)~\cite{BH-09}. Basically, LA takes advantage of (global) static analysis to prune  path prefixes 
which cannot reach any uncovered branches (it could also be adapted for labels).  
%
On $\Pg^{\star}$,  {\sc idl-1} is a very specific (but cheap) case of LA  while {\sc idl-2} is not: LA will prune all ``label paths'' pruned by {\sc idl-1} plus other normal paths  leading   only to already covered labels,   
while  {\sc idl-2} will prune several ``label paths'' at once thanks to  dynamic analysis.

\myparagraph{Automation of mutation testing.} Mutation coverage \cite{DLS-78,OU-01} has been established  as a powerful criteria through several experimental studies~\cite{ABL-05,OU-01}.
 Yet, it is  very difficult to automatize. Even mutation score computation is expensive in practice if not done wisely. Weak mutations \cite{Howden-82} relax mutation coverage  by 
abandoning the ``propagation step'', making \WM\ easier to compare with standard criteria and easier to test for. \WM\ has been experimentally proved to be almost equivalent to 
strong mutations~\cite{OL-94}, and from a theoretical point of view \WM\ subsumes many other criteria \cite{OV-96}. 

The few existing symbolic methods for mutation-based  ATG are based on the encoding proposed by Offutt {\it et al.} and  have already been discussed \cite{DO-91,PMK-10,PM-11}. 
The Mutation Schemata technique \cite{UOH-93} was   originally     developed in order to factorize the compilation costs of hundreds of similar mutants.   
%
%
%
%
Static analysis has been proposed for the  ``equivalent mutant detection'' problem~\cite{OC-94,NW-12} in a way similar  to what is  sketched in Section~\ref{sec:framework}.   

The side-effect free fragment of (atomic) \WM\ presented in this paper  seems to be a sweet spot of mutation testing: it is amenable to efficient automation 
and still very expressive. It is left as future work to identify  if something essential is lost within  this fragment. 
Finally,  our encoding of \WM\ into  \LC\ is orthogonal  
to and can be combined with some of the many techniques developed for efficient mutation testing, such as operator reduction \cite{ORZ-93,WM-95} or smart use of operators \cite{JKS-12a}.


 




%

\section{Conclusion}  \label{sec:conclusion}

We have defined label coverage, a new testing criterion which appears to be both expressive and    amenable to efficient automation.   
Especially, we have shown  that DSE can be extended for label coverage in a black-box manner with only a slight overhead, thanks 
to {\it tight instrumentation} and {\it iterative label deletion}. Experiments show that these two optimisations yield significant improvements.

This work bridges part of the gap between symbolic ATG techniques and coverage criteria. On the one hand, we show that  
DSE techniques  can be cheaply extended to support more advanced testing criteria, including side-effect free weak mutations.  
 On the other hand, we identify a powerful criterion amenable to  efficient automation, both in terms of ATG and coverage score computation.

Future work comprises better delimitation of the expressiveness of labels, designing DSE optimisations geared towards label coverage and carrying out  more thorough experimental evaluations.






\bibliographystyle{abbrv} 

\begin{thebibliography}{1}


\bibitem{ABL-05}  
 J.~H.~Andrews, L.~C.~Briand, Y.~Labiche: 
 \newblock  Is mutation an appropriate tool for testing experiments? 
 \newblock In: ICSE 2005. IEEE 




\bibitem{AO-08} 
 P. Ammann, A. J. Offutt:
\newblock Introduction to software testing.
\newblock Cambridge University Press, New York (2008) 









\bibitem{BH-09} 
S.~Bardin and P.~Herrmann. 
\newblock  Pruning the search space in path-based test generation. 
\newblock In \emph{IEEE ICST 2009}.  IEEE.






\bibitem{BH-11} 
S.~Bardin, P.~Herrmann.  
\newblock    OSMOSE: Automatic Structural Testing of Executables.  
\newblock  Softw. Test., Verif. Reliab. 21(1): 29-54(2011) 







\bibitem{CDE-08} 
C.~Cadar, D.~Dunbar, D.~Engler:
\newblock KLEE: Unassisted and Automatic Generation of High-Coverage
               Tests for Complex Systems Programs. 
\newblock In: OSDI 2008. Usenix Association (2008)                               




\bibitem{CGPDE-06}
C.~Cadar, V.~Ganesh,
               P.~M.~Pawlowski,
               D.~L.~Dill and
               D.~R.~Engler.
\newblock EXE: automatically generating inputs of death.
\newblock In \emph{CCS 2006}. ACM.














\bibitem{CKGJ-10} 
 O.~Chebaro, N.~Kosmatov, A.~Giorgetti, J.~Julliand: 
 \newblock  Combining Static Analysis and Test Generation for C Program Debugging. 
 \newblock In: TAP 2010. 




  \bibitem{CKL-04} 
 E.~M.~Clarke, D.~Kroening, F.~Lerda: 
  \newblock A Tool for Checking ANSI-C Programs. 
  \newblock In: TACAS 2004. Springer (2004)







\bibitem{DLS-78}  
R.~A. DeMillo, R.~J.~Lipton, A.~J.~Perlis:
\newblock Hints on test data selection: Help for the Practicing Programmer. 
\newblock Computer, 11(4), 34-41 

\bibitem{DO-91}  
 R. A. DeMillo, A. J. Offutt: 
 \newblock  Constraint-Based Automatic Test Data Generation. 
 \newblock  IEEE Trans. Software Eng. 17(9), 1991




\bibitem{GKS-05} 
P.~Godefroid, N.~Klarlund and K.~Sen.
\newblock DART: Directed Automated Random Testing.
\newblock In \emph{PLDI'2005}. ACM.






 \bibitem{GLM-08} 
 P.~Godefroid,  M.~Y.~Levin and D.~Molnar.
 \newblock Automated Whitebox Fuzz Testing.
 \newblock In \emph{NDSS 2008}.


\bibitem{GLM-08b} 
P. Godefroid, M. Y. Levin,  D. Molnar:
\newblock Active property checking.
\newblock In: EMSOFT 2008. ACM, New York (2008)                          


 \bibitem{GLM-12} 
 P.~Godefroid, M.~Y.~Levin, D.~A.~Molnar: 
  \newblock SAGE: whitebox fuzzing for security testing. 
  \newblock Commun. ACM 55(3): 40-44 (2012)






 \bibitem{Howden-82} 
W.~E.~Howden:
 \newblock  Weak mutation testing and completeness of test sets.  
 \newblock In: IEEE Transactions on Software Engineering, 8(4). 1982 




 \bibitem{JKS-12a}  
 R.~Just, 
               G.~M.~Kapfhammer and
               F.~Schweiggert. 
 \newblock Do Redundant Mutants Affect the Effectiveness and Efficiency
               of Mutation Analysis? 
 \newblock In: ICST 2012. IEEE (2012)






 \bibitem{KA-96} 
 B.~Korel, A.~M.~Al-Yami: 
 \newblock Assertion-Oriented Automated Test Data Generation. 
 \newblock In: ICSE 1996. IEEE (1996)



\bibitem{Kin-70}
J.~C.~King.
\newblock Symbolic execution and program testing.
\newblock Communications of the ACM, 19(7), july 1976.




 \bibitem{MOK-06} 
 Y. S. Ma, A. J. Offutt, Y. R. Kwon: 
  \newblock  MuJava: a mutation system for java. 
  \newblock In: ICSE 2006. ACM (2006)



 \bibitem{NW-12}   
 S.~Nica, F.~Wotawa: 
 \newblock Using Constraints for Equivalent Mutant Detection. 
 \newblock In: workshop Formal Methods in the Development of Software (2012)





 \bibitem{OC-94}  
 A. J. Offutt, W. M. Craft:  
 \newblock Using Compiler Optimization Techniques to Detect Equivalent Mutants. 
 \newblock Softw. Test., Verif. Reliab. 4(3). 1994



  \bibitem{Offutt-92}  
 A. J. Offutt: 
  \newblock Investigations of the Software Testing Coupling Effect. 
  \newblock ACM Trans. Softw. Eng. Methodol. 1(1): 5-20 (1992)







 \bibitem{OL-94}  
 A. J. Offutt, S. D. Lee: 
 \newblock  An Empirical Evaluation of Weak Mutation.  
 \newblock IEEE Trans. Software Eng. 20(5): 337-344 (1994) 







\bibitem{ORZ-93} 
A. J. Offut, G. Rothermel, C. Zapf:
\newblock An experimental evaluation of selective mutation. 
\newblock In: ICSE 1993. IEEE Press, Los Alamitos (1993)  




\bibitem{OU-01} 
 A. J. Offutt, R. H. Untch:  
 \newblock Mutation 2000: uniting the orthogonal. 
 \newblock In: Mutation testing for the new century.  Kluwer Academic Publisher, 2001


 \bibitem{OV-96} 
 A. J. Offutt, J. Voas: 
 \newblock Subsumption of Condition Coverage Techniques by Mutation Testing. 
 \newblock Tech.~Report ISSE-TR-96-01, Dpt.~ Information and Software System Engineering, George Mason Univ., 1996



\bibitem{Necula02cil} 
G. C. Necula and S. Mcpeak and S. P. Rahul and W. Weimer
\newblock CIL: Intermediate language and tools for analysis and transformation of C programs.
\newblock In: CC 2002. Springer (2002) 




\bibitem{PM-09} 
 M.~Papadakis, N.~Malevris: 
 \newblock An Effective Path Selection Strategy for Mutation Testing. 
 \newblock In:   APSEC 2009. IEEE (2009)



 \bibitem{PM-10} 
M.~Papadakis, N.~Malevris:  
 \newblock  Automatic Mutation Test Case Generation via Dynamic Symbolic Execution. 
 \newblock In:  ISSRE 2010. IEEE (2010)


 

 \bibitem{PM-11} 
 M.~Papadakis, N.~Malevris: 
 \newblock Automatically performing weak mutation with the aid of symbolic execution, concolic testing and search-based testing. 
 \newblock Software Quality Journal 19(4), 2011




\bibitem{PMK-10} %
  M.~Papadakis, N.~Malevris, M.~Kallia:  
 \newblock Towards Automating the Generation of Mutation Tests. 
 \newblock In: workshop AST 2010 (with ICSE 2010).  




\bibitem{SMA-05} 
K.~Sen, D.~Marinov and G.~Agha.
\newblock CUTE: A Concolic Unit Testing Engine for C.
\newblock In:  ESEC/FSE 2005. ACM (2005)




\bibitem{TH-08}   
N.~Tillmann and  J.~de~Halleux.
\newblock  Pex-White Box Test Generation for .NET.
\newblock  In: TAP 2008. Springer (2008)   





 \bibitem{UOH-93}  
 R.~H. Untch, A. J. Offutt, M. J.  Harrold: 
 \newblock Mutation Analysis Using Mutant Schemata. 
 \newblock In: ISSTA. ACM (1993)

	


\bibitem{WM-95} 
W. E. Wong, A. P. Mathur:
\newblock Reducing the cost of mutation testing: An emprirical study.
\newblock In: Journal of Systems and Software, 31(3), 185-196.  


\bibitem{WMM-04a} 
N.~Williams, B.~Marre and P.~Mouy.
\newblock On-the-Fly Generation of K-Path Tests for C Functions.
\newblock In: ASE 2004. IEEE (2004)





 \bibitem{XTHS-09} 
T.~Xie, N.~Tillmann, P.~de~Halleux and W.~Schulte.
 \newblock Fitness-Guided Path Exploration in Dynamic Symbolic Execution. 
 \newblock In: DSN 2009. IEEE (2009) 






	

















 










 


 



 



 


 




\end{thebibliography}



\end{document}